\documentclass{article}
\usepackage[margin=1in]{geometry}
\usepackage{amsmath}
\usepackage{amsfonts}
\usepackage{amssymb}
\usepackage{amsthm}
\usepackage{graphicx}
\usepackage{xcolor}
\usepackage{setspace}
\usepackage{physics}
\usepackage{tikz}
\usepackage{mathdots}
\usepackage{yhmath}
\usepackage{cancel}
\usepackage{color}
\usepackage{siunitx}
\usepackage{array}
\usepackage{multirow}
\usepackage{tabularx}
\usepackage{extarrows}
\usepackage{booktabs}
\usetikzlibrary{fadings}
\usetikzlibrary{patterns}
\usetikzlibrary{shadows.blur}
\usetikzlibrary{shapes}
\newtheorem{prop}{Proposition}

\newcommand{\R}{\mathbb{R}}

\newcommand{\cd}{\xrightarrow[]{d}}

\usepackage{authblk}

\title{Can the potential benefit of individualizing treatment be assessed using trial summary statistics alone?}

\author[1]{Nina Galanter}

\affil[1]{\footnotesize   Department of Biostatistics, University of Washington, Seattle, USA}

\author[1]{Marco Carone}

\author[2]{Ronald C. Kessler}
\affil[2]{\footnotesize Department of Health Care Policy, Harvard Medical School, Boston, USA} 

\author[3]{Alex Luedtke}
\affil[3]{\footnotesize Department of Statistics, University of Washington, Seattle, USA} 

\begin{document}

\maketitle

\begin{abstract}
   Individualizing treatment assignment can improve outcomes for diseases with patient-to-patient variability in comparative treatment effects. When a clinical trial demonstrates that some patients improve on treatment while others do not, it is tempting to assume that treatment effect heterogeneity exists. However, if variability in response is mainly driven by factors other than treatment, investigating the extent to which covariate data can predict differential treatment response is a potential waste of resources. Motivated by recent meta-analyses assessing the potential of individualizing treatment for major depressive disorder using only summary statistics, we provide a method that uses summary statistics widely available in published clinical trial results to bound the benefit of optimally assigning treatment to each patient. We also offer alternate bounds for settings in which trial results are stratified by another covariate. We demonstrate our approach using summary statistics from a depression treatment trial. Our methods are implemented in the \texttt{rct2otrbounds} R package, which is available at \texttt{https://github.com/ngalanter/rct2otrbounds}.
\end{abstract}
\vspace{24pt}

\section{Introduction}

Tailoring treatment based on patient characteristics is a promising way to improve outcomes in major depressive disorder (MDD). Many alternative treatments exist for MDD, each with evidence of effectiveness for some but not all patients, but none is reliably superior to the others. Though numerous significant predictors of differential response across treatments have been reported in the literature, none of them are overwhelmingly powerful \cite{simon2010personalized,webb2019personalized}. Trial and error, somewhat guided by basic insights about differential effectiveness related to patient characteristics, is the core approach to treatment selection. Most patients will need to try multiple treatments to achieve stable remission. 

This pattern has led clinical researchers to hypothesize that characteristics of the patient and their symptoms define some as-yet-unknown MDD subtypes that respond differently to available treatments and lead to substantial heterogeneity in the comparative effectiveness of treatment alternatives. If true, discovering these characteristics would lead to improved treatment outcomes by allowing the right treatments to be assigned to patients quickly and without the current trial and error process. Several studies have developed individualized treatment rules (ITRs) for MDD or MDD subtypes based on the discovery of interactions between baseline patient characteristics and treatment types in predicting outcomes. Most of these ITRs are weak, however, either because the samples were too small to support stable estimation or because limited predictors were used in developing the ITRs \cite{zhao2012estimating, kessler2021individualized, derubeis2014personalized}. These repeated failures lead to the question: is it possible to develop ITRs that are substantially beneficial in guiding MDD treatment selection? If strongly beneficial ITRs exist, the right predictors simply remain to be found. However, if strongly beneficial ITRs do not exist, further efforts aimed at finding the right predictors would be a futile enterprise. Of course, the existence of variability in MDD outcomes is undeniable, but that does not mean that systematic treatment effect heterogeneity exists. The outcome variation could be due to factors such as measurement error, unrelated changes in patient disease trajectory, and
variability between patients unrelated to treatment \cite{winkelbeiner2019evaluation}. 

 Several groups have proposed that investigating the ratio of treatment arm-specific outcome variances across all arms in a trial determines the potential for developing ITRs. The claim is that if the variance ratio is close to one, there is little evidence for meaningful benefits from individualized treatment \cite{maslej2020individual,munkholm2020individual, ploderl2019chances,volkmann2020treatment}. However, this recommendation is incorrect, as a given variance ratio can be consistent with many potential levels of benefit from individualized treatment. Volkmann \cite{volkmann2020relationship} recently derived a bound on the treatment effect heterogeneity based on the variance ratio; this bound demonstrates the range of heterogeneity values possible for any variance ratio.  However, the reported bound does not provide insight into the relationship between treatment effect heterogeneity and the benefit of individualizing treatment. The bound is also not tight for outcome measures restricted within a given range. Given this, it would be advantageous to have a measure that can be easily calculated from study summary statistics and provides guidance on the potential aggregate gain from individualizing treatment. 

In this paper, we directly link trial summary statistics to the benefit of individualized treatment by providing upper and lower bounds
on the average benefit gained from using an optimal treatment rule as opposed to the best population-level (i.e., unindividualized) treatment. The only summary statistics needed to compute these bounds are the means and variances in each treatment arm, all of which can be estimated from a meta-analysis of studies on a given treatment or class of treatments \cite{maslej2020individual,munkholm2020individual,winkelbeiner2019evaluation}. We provide a general bound and refined bounds for binary and otherwise bounded discrete outcomes. We argue that the latter bounds are as tight as possible given these summary statistics, in the sense that no better bounds that only use these summary statistics can be found. With these bounds, investigators can provide a range of possible values for the gains realizable by assigning treatment according to an optimal rule.  
We also provide improved bounds for use when an investigator has access to arm-specific means and variance estimates stratified by a covariate, such as sex or disease severity.

\section{Context and use cases for our results}

It is noteworthy that the approach presented here derives bounds on the benefit of individualizing treatment based on summary statistics that can be used whether or not individual-level trial data are accessible by the investigator. When individual-level data are not available, the derived bounds provide insight as to whether there is potential benefit to getting access to such data, either from existing trials that collected data on covariates that may modify the treatment effect or by conducting a new study. If, instead, individual-level data are available, then the bounds presented here can aid in deciding whether it is worth putting in the effort needed to estimate an optimal treatment rule based on the available covariates.

If our bounds on the benefit of individualizing treatment suggest that it could be worthwhile to estimate an optimal treatment rule and evaluate its benefit, there are many available methods. For example, the \texttt{DynTxRegime} R package \cite{dyntxregime} implements methods for individualized treatment regimes including Q-learning \cite{murphy2005generalization}, interactive Q-learning \cite{dyntxregime}, outcome weighted learning \cite{zhao2012estimating}, and residual weighted learning \cite{zhou2017residual}. The \texttt{DTRreg} R package \cite{dtrreg} implements the weighted least squares \cite{chakraborty2013statistical,wallace2015doubly} and G-estimation \cite{chakraborty2013statistical,robins2004optimal} methods. The \texttt{sg} R package \cite{vanderweele2019selecting} provides a means to estimate individualized treatment rules via super learning \cite{luedtke2016super} and evaluate their benefit via targeted minimum loss-based estimation \cite{van2015targeted}.

Figure~\ref{fig:flowchart} presents a summary of how the results in this paper can be used.

\begin{figure}[ht]
\centering
\includegraphics[scale = 0.7]{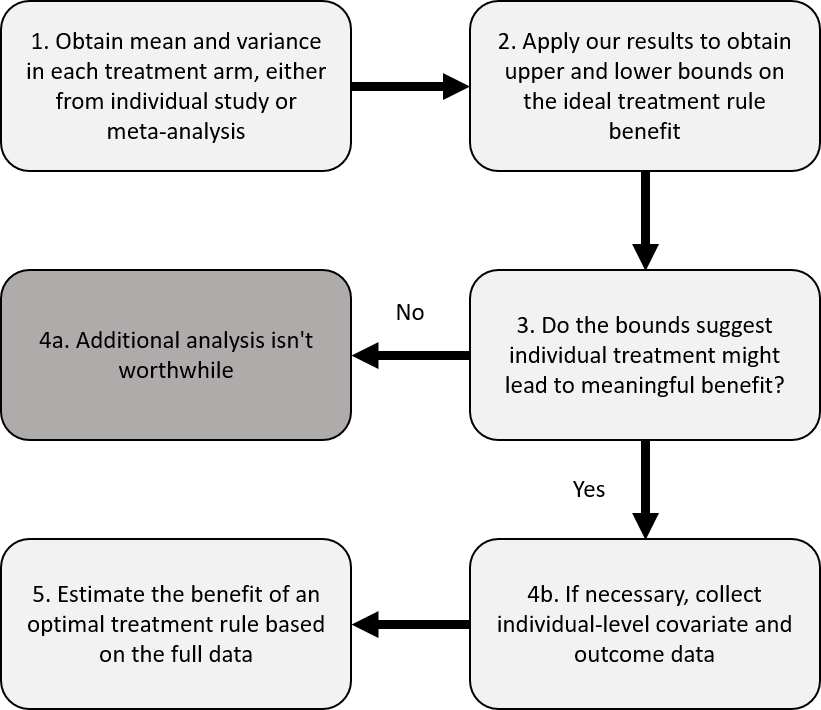}
\caption{Flowchart of the process of applying our bounds on the treatment rule benefit.}
\label{fig:flowchart}
\end{figure}

\section{Results}\label{bounds_section}

\subsection{Setting and notation}

We begin by establishing our setting of interest and notation. Throughout, we consider the case where there are two treatments, treatment 1 and treatment 0, which we call ``treatment" and ``control", respectively. Naturally, all of our results are also applicable when both treatment 1 and treatment 0 are active. For a given individual $i$, we denote the outcomes under treatment and control as $Y^1_{i}$ and $Y^0_{i}$, respectively. We will use $Y^1$ and $Y^0$ to refer to the potential outcomes of a generic individual randomly drawn from the target population. Throughout, we assume that larger outcome values are preferable to smaller values, so that the goal of treatment is to increase the outcome. If smaller outcome values are preferred, the outcomes should be transformed appropriately. For example, if the outcome is an indicator of having an adverse outcome, the indicator of not having an adverse outcome could be considered instead; our bounds would then remain the same. We also denote the individual-level treatment effect variable by $\Delta=Y^1-Y^0$, which we refer to as the treatment effect for brevity. Finally, we define the variance ratio described earlier as $\nu^2=\mathrm{var}(Y^1)/\mathrm{var}(Y^0)$. 

In Sections \ref{bounds_section} and \ref{extensions_section}, all of our results involve the true, population-level average treatment effects and variances of the treatment and control potential outcomes. Section \ref{estimation_section} describes how to estimate our bounds and create corresponding confidence intervals based on the summary statistics available from a completed trial. Our methods are implemented in the rct2otrbounds R package, which is available at \texttt{https://github.com/ngalanter/rct2otrbounds}.

\subsection{Bounds on treatment effect heterogeneity}

The treatment effect heterogeneity,  $\mathrm{var}(\Delta)$, can be bounded based only on knowledge of the variances in the two arms, namely $\mathrm{var}(Y^0)$ and $\mathrm{var}(Y^1)$. Recalling that $\nu^2=\mathrm{var}(Y^1)/\mathrm{var}(Y^0)$, the bound writes as
\begin{equation}\label{teh_bound}
 \mathrm{var}(Y^0)(\nu-1)^2\leq \mathrm{var}(\Delta)\leq \mathrm{var}(Y^0) (\nu+1)^2.
\end{equation} This two-sided bound on the treatment effect heterogeneity was derived by \cite{volkmann2020relationship}; we provide an alternative derivation in Appendix  \ref{sup_var_bound}.

The lower bound attains a minimum value of zero when the outcome variances are equal across trial arms, coinciding with the intuition that such a scenario is consistent with little or no treatment effect heterogeneity. However, even when these variances are equal, the upper bound on $\mathrm{var}(\Delta)$ is four times the outcome variance in the control arm. For $\nu$ close to zero, both the upper and lower bounds approach a minimum of $\mathrm{var}(Y^0)$. In general, the difference between the bounds is $4\nu \mathrm{var}(Y^0)$, and thus, as $\nu$ becomes larger, the range of heterogeneity values consistent with the bounds increases.

The true treatment effect heterogeneity attains the lower bound when the potential outcomes are perfectly correlated between trial arms, meaning that patients who improve on treatment will also improve on control and vice versa. The true treatment effect heterogeneity attains the upper bound when the potential outcomes are instead perfectly negatively correlated, meaning that patients who improve on treatment will worsen on control and those that worsen on control will improve on treatment. Thus, for any outcome measure for which perfect negative and positive correlations are theoretically possible, no uniformly better bound exists that uses only summary statistics.

When the outcomes are bounded within a range $[m,M]$, as is the case with assessment instruments, we can construct a tighter bound using a result from H\"{o}ssjer and Sj\"{o}lander \cite{hossjer2022sharp}. Suppose that, in addition to the arm-specific variances, the arm-specific means $E(Y^0)$ and $E(Y^1)$ are available. In this case, we obtain the bound
\begin{align}\label{teh_bound_bounded}
\mathrm{var}(Y^0)+\mathrm{var}(Y^1)-s^-\leq \mathrm{var}(\Delta)\leq \mathrm{var}(Y^0)+\mathrm{var}(Y^1)+s^+\ ,
\end{align}where we define $s^-=2\,\mathrm{min}\left[\sqrt{\mathrm{var}(Y^0)\mathrm{var}(Y^1)},\{M-E(Y^0)\}\{E(Y^1)-m\},\{E(Y^0)-m\}\{M-E(Y^1)\}\right]$ and $s^+=2\,\mathrm{min}\left[\sqrt{\mathrm{var}(Y^0)\mathrm{var}(Y^1)},\{E(Y^0)-m\}\{E(Y^1)-m\},\{M-E(Y^0)\}\{M-E(Y^1)\}\right]$.
In practice, this bound is more likely to result in a meaningful improvement when the outcome standard deviation in each arm is large relative to the possible range of the data. In the special case in which the outcome is binary, simpler expressions are available for $s^-$ and $s^+$, namely $s^-=2\mathrm{min}\{E(Y^1),E(Y^0)\}+2E(Y^1)E(Y^0)$ and $s^+=2\mathrm{min}\{0,E(Y^1)+E(Y^0)-1\}+2E(Y^1)E(Y^0)$.

\subsection{Main Result: Bounds on optimal treatment rule benefit}\label{sec:bdsOptRuleBenefit}

The next result we present can be computed when there is knowledge of the arm-specific mean outcomes and variances, namely $E(Y^0)$, $E(Y^1)$, ${\rm var}(Y^0)$, and ${\rm var}(Y^1)$. This result provides tight upper and lower bounds on the best possible benefit of individualizing treatment based on covariates. To do this, we consider what we call an ideal treatment rule, which assigns treatment if $Y^1>Y^0$ and control if instead $Y^0>Y^1$.  If $Y^0=Y^1$, an arbitrary choice of treatment (e.g., always assign treatment) can be made. Since an ideal rule perfectly discriminates the preferred treatment for each individual, the mean outcome under an ideal rule is always at least as large as the mean outcome under an optimal treatment rule that assigns treatment based only on available covariates. Moreover, in an extreme scenario where enough covariates are available to uniquely identify each individual in the population and perfectly predict their ideal treatment, the optimal rule based on these covariates is itself an ideal rule.

Define the expected outcome under a best treatment and an ideal rule as $\mu_T=\max\{E(Y^0),E(Y^1)\}$ and  $\mu_O=E\{\max(Y^0,Y^1)\}$. We will first look at the case when the outcome of interest is discrete and bounded, as  is commonly the case (e.g., scored questionnaire). Using the partial identifiability framework \cite{balke1994counterfactual,duarte2021automated}, the bounds are optimal values of linear programs. Specifically, for an outcome with $k$ possible values $y_1,y_2,\ldots,y_k$, we solve for the value of the $k^2$ probabilities $p_{jr}=P(Y^0=y_j,Y^1=y_r)$, $j,r=1,2,\ldots,k$, defining the joint distribution of $Y^1$ and $Y^0$ that maximizes the benefit of an ideal rule, under the constraint that the mean outcome and outcome variance in each arm must equal their observed value. Without loss of generality, we assume that the average outcome is higher in the treatment arm. If the average outcome is higher in the control arm, we can switch the labels in the following equation. An upper bound on $\mu_O-\mu_T$ is given by the optimal value of the following linear program:
\begin{eqnarray}
&\textnormal{maximize} &\sum_{r=1}^k\sum_{j=1}^k (y_r-y_j)p_{rj}I(y_j < y_r) \label{lp_bound}\\[0.5em]
&\textnormal{subject to } &\min_{r,j} p_{rj}\ge 0,\hspace{1em} \sum_{r=1}^k\sum_{j=1}^k p_{rj} = 1,\hspace{1em}
\sum_{r=1}^k \sum_{j=1}^k y_j p_{rj} = E(Y^1)\ ,\hspace{1em}
\sum_{r=1}^k\sum_{j=1}^k y_r p_{rj}  = E(Y^0)\ , \nonumber\\
& &\sum_{r=1}^k \sum_{j=1}^k y_j^2 p_{rj} = \mathrm{var}(Y^1) + \{E(Y^1)\}^2\ ,\hspace{1em}\sum_{r=1}^k\sum_{j=1}^k y_r^2 p_{rj} =\mathrm{var}(Y^0) + \{E(Y^0)\}^2\ . \nonumber
\end{eqnarray}
This upper bound is tight in the sense that there necessarily exists a joint distribution of $Y^0$ and $Y^1$ that is consistent with the observed summary statistics and under which $\mu_O-\mu_T$ equals the bound. 
A tight lower bound can be found by minimizing the objective function subject to the same constraints. The described linear programs can be readily solved with existing software such as the `lpsolve' package in R \cite{lpsolve, r}. 

As an ideal rule always results in mean outcomes at least as large as those under optimal rules based on available covariates, \eqref{lp_bound} also provides an upper bound on the difference between expected outcomes under an optimal rule and a best single treatment. However, whether the lower bound will be lower than the value of an optimal rule depends on the extent to which the covariates used to create the rule are predictive of the individual-level treatment effect. 

Unsurprisingly, the lower bound on the benefit of individualizing treatment will often be quite close to zero, and in many cases it will be exactly zero. In fact, the lower bound can be zero even when \eqref{teh_bound} or \eqref{teh_bound_bounded} suggests that the variance of the treatment effect, ${\rm var}(\Delta)$, must be greater than zero. To see why this is possible, note that when the average treatment effect $E(\Delta)$ is far from zero, it is possible that there is individual level heterogeneity in treatment effect, that is, ${\rm var}(\Delta)>0$, and yet the ideal rule would still give everyone the same treatment. Nevertheless, the lower bound can be meaningfully larger than zero in certain cases. This will happen, in particular, when the lower bound on the variance of the treatment effect is large and yet the average treatment effect is zero or nearly zero. In such cases, there will necessarily be some individuals for whom receiving treatment is optimal and others for whom receiving control is optimal. Hence, there will necessarily be a benefit from individualization in such cases, and so the tight lower bound derived by linear programming will reflect that.

When the outcome is binary, the linear programming bounds simplify to the following closed-form bounds:
\begin{equation}\label{rule_bound_binary}\begin{matrix}
0\leq \mu_O - \mu_T\leq \mathrm{min}\{E(Y^0),1-E(Y^1)\} & \text{if }E(Y^0)\leq E(Y^1),\\
0\leq \mu_O - \mu_T\leq \mathrm{min}\{1-E(Y^0),E(Y^1)\} & \text{if }E(Y^0)>E(Y^1).
\end{matrix}
\end{equation} As can be seen from the lower bounds above, in the binary case it is always possible that there is no benefit to individualizing treatment.  

It is also possible to provide general closed-form upper bounds on the ideal treatment rule benefit. Using the bound on the treatment effect heterogeneity given in the previous section, we can bound the difference in average outcomes under an ideal treatment rule and a best single treatment. We can define $B^+$ as the relevant upper bound on the treatment effect variance, so that the result uses \eqref{teh_bound} for unbounded outcomes and \eqref{teh_bound_bounded} for outcomes with a bounded range (e.g., binary outcomes or questionnaire scores). The bound
\begin{equation}\label{rule_bound}
\mu_O-\mu_T\leq \frac{1}{2}\sqrt{B^+ +\{E(\Delta)\}^2}
\end{equation}
holds with $E(\Delta)$ representing the average treatment effect. Unlike the linear programming bound, this bound has a closed-form expression and applies even if the outcome can take infinitely many values. However, when the linear programming bound does apply, it will always be at least as tight as the one above.

\section{Extension: incorporating a stratification variable}\label{extensions_section}

Often, trial summary statistics are reported not only for each trial arm but also by levels of a baseline stratification variable.
Common examples of stratification factors include sex \cite{legates2019sex} and baseline indicators of disease severity and complexity \cite{bartova2019results}. We focus on the case where the stratification factor $C$ can take possible values $c_1,c_2,\ldots,c_k$. We derive improvements on the bounds from Section~\ref{bounds_section} in the case where there is knowledge of the marginal probability $\mathrm{pr}(C=c_s)$ of belonging to each subgroup and of the same summary statistics as were used in Section~\ref{bounds_section}, but within subgroups defined by $C$.

Given lower and upper bounds $B^-_s$ and $B^+_s$ for treatment effect heterogeneity within stratum $C=c_s$, 
the treatment effect heterogeneity, $\mathrm{var}(\Delta)$, can be bounded as
\begin{equation}\label{strata_bound}
r_C+\sum_{s=1}^k \mathrm{pr}(C=c_s) B_s^-\leq \mathrm{var}(\Delta)\leq r_C+\sum_{s=1}^k \mathrm{pr}(C=c_s) B_s^+\ ,
\end{equation}where we define $r_C=\sum_{s=1}^k \mathrm{pr}(C=c_s)\left\{E(\Delta\mid C=c_s)-E(\Delta)\right\}^2$. These bounds are each weighted sums of two terms across strata. The first term is the variance of the stratified average treatment effect around the overall average expected treatment effect. The second term bounds the average treatment effect variance conditional on $C=c_s$.

To obtain tight bounds on the benefit of individualizing treatment, we decompose this benefit into two terms. The first corresponds to the benefit of knowing a patient's stratum of $C$ when deciding their treatment relative to not knowing any individual-level information. The second corresponds to the further benefit of knowing all individual-level information relevant to making a treatment decision, such that it is possible to assign the ideal treatment. In other words, the tight bounds on the benefit of individualizing treatment leverage the decomposition
\begin{align}
    \mu_O-\mu_T = (\mu_C-\mu_T) + (\mu_O-\mu_C), \label{eq:stratifiedDecomp}
\end{align}
where $\mu_C$ denotes the population average outcome under the rule that assigns each patient the treatment that performs best in their stratum of $C$. Note that $\mu_C-\mu_T=\sum_{s=1}^k \mathrm{pr}(C=c_s)\max\{E(Y^0|C=c_s),E(Y^1|C=c_s)\}-\mu_T$. Hence, the value of the first term above can be computed directly from strata-specific summary information. In contrast, the precise value of the second term cannot be learned from such coarse information. Instead, bounds on this quantity can be derived by applying the linear programming technique in Section~\ref{sec:bdsOptRuleBenefit} within each stratum of $C$. In particular, $\sum_{s=1}^k\mathrm{pr}(C=c_s)b^-_{c_s}\le \mu_O-\mu_C\le \sum_{s=1}^k\mathrm{pr}(C=c_s)b^+_{c_s}$, where $b^-_{c_s}$ and $b^+_{c_s}$ are optimal values of linear programs presented in Appendix \ref{app:stratifiedlp}. Returning to \eqref{eq:stratifiedDecomp}, the final bounds on the benefit of individualization are given by
\begin{align*}
    (\mu_C-\mu_T)+\sum_{s=1}^k\mathrm{pr}(C=c_s)b^-_{c_s}\leq\mu_O-\mu_T\leq (\mu_C-\mu_T)+\sum_{s=1}^k\mathrm{pr}(C=c_s)b^+_{c_s}.
\end{align*}
These bounds are at least as tight as the linear programming bound in Section~\ref{sec:bdsOptRuleBenefit}, which did not use stratum-level information.

The equation in \eqref{eq:stratifiedDecomp} also implies a closed-form upper bound on the benefit of an ideal treatment rule. In particular, given upper bounds $B^-_s$ and $B^+_s$ for treatment effect heterogeneity within stratum $C=c_s$,
\begin{equation}
\mu_O-\mu_T\leq \mu_C-\mu_T +\frac{1}{2}\sum_{s=1}^k \mathrm{pr}(C=c_s)\sqrt{B_s^+ +\{E(\Delta|C=c_s)\}^2}\ .
\end{equation}
This bound applies regardless of how many values the outcome can take. However, when the outcome can only take finitely many values so that the linear programming bound applies, that bound is generally tighter than the above.

\section{Estimation and inference for the bounds}\label{estimation_section}

The bounds presented thus far involve the true (stratum-specific) average treatment effect and potential outcome variances. In practice, these quantities are unknown and must be estimated. Based only on summary statistics, we can construct consistent estimators for the bounds and asymptotically conservative confidence bounds for the estimators.

Again, we begin by defining notation. We denote by $\bar{Y}^1$ and $\bar{Y}^0$ the sample means of  outcomes in the treatment and control arms, respectively, and by $\bar{\Delta}=\bar{Y}^1 - \bar{Y}^0$ the sample average treatment effect. We also denote by $s_1$ and $s_0$ the sample standard deviations of outcomes in the treatment and control arms, respectively. The general bound on the treatment effect heterogeneity in \eqref{teh_bound}, the closed-form bound on the benefit of the optimal rule in \eqref{rule_bound}, and the rule benefit bound for binary outcomes in \eqref{rule_bound_binary} can be estimated consistently by replacing $E(\Delta)$, $\mathrm{var}(Y^0)$ and $\nu$ by $\bar{\Delta}$,  $s_0^2$ and $s_1/s_0$, respectively, in the provided formulas. The bounds involving a stratifying covariate in Section \ref{extensions_section} can be consistently estimated using a similar approach implemented within each stratum. The linear programming bound on the rule benefit for discrete bounded outcomes can be estimated by replacing $E(Y^0)$, $E(Y^1)$, $\mathrm{var}(Y^0)$ and $\mathrm{var}(Y^1)$  by $\bar{Y}^0$, $\bar{Y}^1$, $s_1^2$ and $s_0^2$, respectively, in the constraints. The estimated bound will tend to the true bound if the range of the outcome is positive (or the outcome is transformed to have a positive range) --- see Corollary 3.1 of \cite{bereanu1976continuity}.

As for uncertainty quantification, we consider two settings with different levels of available information in Appendix  \ref{sup_inference}. First, we consider scenarios in which only the sample mean and variance in each arm of the trial are available. In such cases, there is insufficient information to construct standard confidence intervals for the treatment effect heterogeneity bound or for the bound on the ideal rule benefit. Nevertheless, if the outcome is bounded, asymptotically conservative confidence intervals can be constructed. The method we propose is applicable, for example, when the outcome is binary or a value on a rating scale. Second, we consider scenarios in which individual trial patient data are available for the outcome but not other covariates (e.g., due to privacy concerns). In such cases, asymptotically valid confidence intervals can be constructed for the bounds on unbounded outcomes in \eqref{teh_bound_bounded} and \eqref{rule_bound}. We provide variance formulas that can be used to construct confidence intervals and hypothesis tests  for these two settings in Appendix  \ref{sup_inference}.

\section{Illustration of bounds using EMBARC trial data}
\label{example_section}

The EMBARC study was a sequential multiple assignment randomized trial in which 296 adults with early-onset (age $<30$ years) recurrent MDD were randomized either to the SSRI sertraline or a placebo, with non-responders undergoing randomization to a second treatment after eight weeks \cite{webb2019personalized}. Outcomes were only reported for participants who completed treatment (115 on sertraline, 123 on placebo). One of the primary outcomes, which we focus on, was the post-treatment 17-item Hamilton Rating Scale for Depression, for which a lower value is better. We focus on treatment response after eight weeks from randomization to examine treatment heterogeneity in sertraline response. 

Although the computation of our bounds only requires access to trial summary statistics, we use the full EMBARC dataset to assess how close our estimated bounds are to the estimates obtained based on the entire dataset. Specifically,  we illustrate the use of our method by using the data to estimate the benefit of an optimal treatment rule based on the covariates collected in the trial. We calculate these estimates using cross-validated targeted minimum loss-based estimation \cite{luedtke2016super}. We restrict our analysis to patients with non-missing outcomes, as missing outcomes would be excluded from the relevant summary statistics. We use single imputation to address missing covariate values; details are provided in Appendix  \ref{sup_data_example}. Importantly, this estimate does not fully capture the potential gains from individualizing treatment in each scenario, as there may be additional unrecorded patient characteristics whose use would lead to improved rules.

The possible range of the 17-item Hamilton Scale is 0 through 52, but at eight weeks, the subjects' scores ranged from 0 through 32. The mean values in the treatment and placebo arms were 10.73 and 11.94, respectively, leading to an estimated average treatment effect of $-1.21$. The sample standard deviations in the treatment and placebo arms are 6.53 and 7.52, respectively. 

Using the linear optimization bounds on the benefit of individualized treatment using only summary statistics, we estimate that the range of possible benefits is $(0,6.43)$ points on the Hamilton scale with a 95\% confidence interval of $(0,12.25)$. Hence, no rule, regardless of its form or the covariates involved, will outperform treating everyone by more than an estimated 6.43 points (97.5\% upper confidence limit: 12.25 points). The trial results are also consistent with little or no gain from individualizing treatment. An analysis using the full, individual-level EMBARC data lends more support to the latter of these possibilities, with a 95\% confidence interval of $(0,0.09)$ for the benefit of tailoring treatment based on the covariates measured in the trial.  Additionally, although in this trial outcome data were not reported stratified by an additional variable, we used the full data to create several bounds that would have been computable were stratified outcome summaries reported. Bounds using binarized sex and age group were not meaningfully different from the overall bound. 

As mentioned in Section~\ref{sec:bdsOptRuleBenefit}, the magnitude of our bounds can change depending on the value of the average treatment effect. To illustrate this, we consider the hypothetical settings where the mean on the placebo arm and the variances on both arms remain the same as those observed in the EMBARC study, but the estimated average treatment effect is either larger or smaller. 
If the estimated average treatment effect had been three times as great as in the actual study, then the bound on the benefit of individualized treatment based on the summary statistics would have been (0, 5.43), with a 95\% confidence interval of (0, 9.83). If there had been no estimated average treatment effect, then the bound on the benefit of individualized treatment based on the summary statistics would have been (0.14, 7.01), with a 95\% confidence interval of (0, 13.46). These results suggest that a higher treatment effect will lead to a lower potential benefit of individualized treatment, whereas no treatment effect leads to a higher potential benefit.

\section{Discussion}

So, can the potential benefit of individualizing treatment be assessed using trial summary statistics alone? 
Yes, to an extent. 
Undoubtedly, the benefit cannot be assessed to a high degree of precision without individual-level data containing treatment effect modifiers. 
However, collecting such data is often time-consuming and expensive, so having a strategy for narrowing down the range of potential benefits beforehand is valuable. 
To this end, we have provided an approach for using summary statistics to give the tightest possible bounds on the benefit of individualizing treatment. 
In certain cases, our approach even makes it possible to determine that there must be some benefit and, if so, to specify a lower bound on that benefit.

In our illustration based on EMBARC data, although the upper bound found is considerably higher than the benefit estimated with the baseline covariates, no tighter bound on the benefit below our linear programming bound can be found without individual-level data.  Furthermore, the discrepancy may stem from the fact that the bound we have derived is on the maximum benefit using all possible covariates, whereas the estimated benefit only accounts for covariates collected in the study, which may not include key treatment effect modifiers.

Several extensions of our methods are possible. Though we focused on the case where there are only two treatments, our linear programming bounds can easily be generalized to settings where the are multiple treatments. Our method can also be extended to observational designs or studies using covariate-based randomization schemes if adjusted estimates are provided or the full data is available. However, in these cases, the arm-specific empirical outcome means and variances no longer provide valid estimates of the means and variances of the potential outcomes under each treatment.

The bounds we have derived are especially helpful for characterizing the potential for personalized treatment for a disease across a broad treatment class. This is the case in our motivating example, where researchers are attempting to determine whether personalized treatment is plausibly beneficial in the general use of SSRI treatments for MDD \cite{maslej2020individual,munkholm2020individual}. Our primary bounds only require arm-specific outcome means and variances. Common means and variances can be estimated across all applicable studies using methods for meta-analysis. Then, our method could be used to assess the general benefit of individualized treatment across studies. Our bounds are not specific to the depression treatment case and could also be used to investigate the potential for benefit from individualized treatment for other diseases.

\section*{Acknowledgements}

This work was supported by the National Institutes of Health (NIH) through award numbers DP2LM013340, R01HL137808, and T32ES015459. NG was partially supported by the National Science Foundation (NSF) Graduate Research Fellowship Program under grant number DGE-1762114. Data used in the preparation of this manuscript were obtained from the National Institute of Mental Health Data Archive (NDA). NDA is a collaborative informatics system created by the NIH to provide a national resource to support and accelerate research in mental health. Dataset identifier: 2199. Any opinions, findings, conclusions, or recommendations expressed in this material are those of the authors and do not necessarily reflect the views of the NIH, NSF, or of the submitters of the original data to NDA.

\bibliographystyle{acm}
\bibliography{teh_bound_lit}

\appendix

\renewcommand{\thesection}{\Alph{section}}

\setcounter{equation}{0}
\renewcommand{\theequation}{S\arabic{equation}}

\renewcommand{\thetable}{S\arabic{table}}
\renewcommand{\thefigure}{S\arabic{figure}}

\section*{Appendices}

Appendix \ref{var_bounds} presents bounds on $\mathrm{var}(\Delta)$ both in the cases where $Y$ is and is not a bounded random variable. Appendix~\ref{app:idealTxBenefit} presents bounds on the benefit of individualizing treatment. Appendix \ref{extra_bounds} presents the form of the linear programs used when stratified summary statistics are available, along with the proofs regarding the validity and relative tightness of our closed-form bound for this case. Appendix \ref{sup_inference} presents methods to obtain inference for the bounds. Appendix \ref{sup_data_example} presents further details on the illustration with data from the EMBARC trial.

\section{Proofs of bounds on treatment effect heterogeneity}\label{var_bounds}

\subsection{General case}\label{sup_var_bound}

\begin{prop}\label{prop:sup_var_bound}
If $\mathrm{var}(Y^0)$ and $\mathrm{var}(Y^1)$ are both finite and $\mathrm{var}(Y^0)>0$, then
\begin{align}\label{eq:tightVarBd}
    \mathrm{var}(Y^0)(\nu-1)^2  \leq \mathrm{var}(\Delta)\leq \mathrm{var}(Y^0)(\nu+1)^2.
\end{align}
\end{prop}
\begin{proof}
Using that $\Delta=Y^1-Y^0$ and the definition of the correlation between two random variables,
\begin{align*}
\mathrm{var}(\Delta) & =\mathrm{var}(Y^1-Y^0)=\mathrm{var}(Y^1)+\mathrm{var}(Y^0)-2\mathrm{cov}(Y^1,Y^0)\\
&=\mathrm{var}(Y^1)+\mathrm{var}(Y^0)-2\mathrm{cor}(Y^1,Y^0)[\mathrm{var}(Y^1)\mathrm{var}(Y^0)]^{1/2}.
\end{align*}
As $\mathrm{var}(Y^1)=\nu^2 \mathrm{var}(Y^0)$, the above rewrites as
\begin{align*}
\mathrm{var}(\Delta)
&=(\nu^2+1)\mathrm{var}(Y^0)-2\mathrm{cor}(Y^1,Y^0)\nu \mathrm{var}(Y^0).
\end{align*}
Dividing both sides by $\mathrm{var}(Y^0)>0$ shows that
\begin{align}
\label{eq_for_diff}
\frac{\mathrm{var}(\Delta)}{\mathrm{var}(Y^0)}=\nu^2+1-2\mathrm{cor}(Y^1,Y^0)\nu.
\end{align}
As correlations fall in $[-1,1]$,
\begin{align*}
\nu^2+1-2\nu \leq \frac{\mathrm{var}(\Delta)}{\mathrm{var}(Y^0)}&\leq \nu^2+1+2\nu.
\end{align*}
The left- and right-hand sides above rewrite as $(\nu-1)^2$ and $(\nu+1)^2$, respectively.
\end{proof}
The bounds in Eq.~\ref{eq:tightVarBd} can be shown to be tight, in the sense that, for any values $\mu_0$ and $\mu_1$ of the arm-specific means $E(Y^0)$ and $E(Y^1)$ and any positive values $\sigma_0^2$ and $\sigma_1^2$ of the arm-specific variances $\mathrm{var}(Y^0)$ and $\mathrm{var}(Y^1)$, each of these bounds can be attained for at least one joint distribution of the counterfactual outcomes $Y^0$ and $Y^1$ that has these arm-specific means and variances. To achieve the lower bound, we can take $(Y^0,Y^1)$ to be a random variable that is uniformly distributed over the doubleton set $\{(\mu_0-\sigma_0,\mu_1-\sigma_1),(\mu_0+\sigma_0,\mu_1+\sigma_1)\}$. To achieve the upper bound, we can instead take $(Y^0,Y^1)$ to be a random variable that is uniformly distributed over the doubleton set $\{(\mu_0-\sigma_0,\mu_1+\sigma_1),(\mu_0+\sigma_0,\mu_1-\sigma_1)\}$.

\subsection{Improved bound for bounded outcomes}\label{sup_var_bound_bounded}

\begin{prop}\label{prop:sup_var_bound_bdd}
For an outcome $Y$ bounded in a range $(m,M)\subset \R$, the following bound holds:
\begin{align}\label{eq:tightVarBdBddY}
\mathrm{var}(Y^1)+\mathrm{var}(Y^0) - s^{-}\le \mathrm{var}(\Delta)\le \mathrm{var}(Y^1)+\mathrm{var}(Y^0) + s^{+},
\end{align}
where
\begin{align*}
    s^-&=2\,\mathrm{min}\left[\sqrt{\mathrm{var}(Y^0)\mathrm{var}(Y^1)},\{M-E(Y^0)\}\{E(Y^1)-m\},\{E(Y^0)-m\}\{M-E(Y^1)\}\}\right], \\
    s^+&=2\,\mathrm{min}\left[\sqrt{\mathrm{var}(Y^0)\mathrm{var}(Y^1)},\{E(Y^0)-m\}\{E(Y^1)-m\},\{M-E(Y^0)\}\{M-E(Y^1)\right].
\end{align*}
\end{prop}

\begin{proof}
The bound follows from the fact that
\begin{align*}
    \mathrm{var}(\Delta) &= \mathrm{var}(Y^1)+\mathrm{var}(Y^0)-2\mathrm{cov}(Y^1,Y^0),
\end{align*}
along with the bound on $\mathrm{cov}(Y^1,Y^0)$ from Theorem 2 of H\"{o}ssjer and Sj\"{o}lander \cite{hossjer2022sharp}. 

\end{proof}

The tightness of the bounds in Eq.~\ref{eq:tightVarBdBddY} follows from the fact that Theorem 2 of H\"{o}ssjer and Sj\"{o}lander provides tight bounds on the covariance between bounded random variables with a given range \cite{hossjer2022sharp} --- see that reference for further information.

\section{Proofs of bounds on the benefit of individualization}\label{app:idealTxBenefit}

\subsection{Closed-form bound on the benefit of individualization}\label{sup_rule_bound}

\begin{prop}
The following is a bound on the benefit of individualizing treatment:
\begin{equation*}
\mu_O-\mu_T\leq \frac{1}{2}\sqrt{B^+ +\{E(\Delta)\}^2},
\end{equation*}
where $B^+$ is the relevant upper bound on the treatment effect variance, so that the result uses Eq.~\ref{teh_bound} for unbounded outcomes and Eq.~\ref{teh_bound_bounded} for outcomes with a bounded range (e.g., binary outcomes or questionnaire scores).

\end{prop}

\begin{proof}
Without loss of generality, assume that $E(Y^1)\geq E(Y^0)$. Let $Z=I\{\Delta\geq 0\}$. Under the optimal rule, there is an expected outcome of
\begin{align*}
\mu_O&= E[Y^1Z+Y^0(1-Z)] =\mathrm{pr}(Z=1)E(Y^1\mid Z=1)+\mathrm{pr}(Z=0)E(Y^0\mid Z=0).
\end{align*}
Under treatment 1, the expected outcome is
\begin{align*}
E(Y^1)&= E[Y^1Z+Y^1(1-Z)]=\mathrm{pr}(Z=1)E(Y^1\mid Z=1)+\mathrm{pr}(Z=0)E(Y^1\mid Z=0).
\end{align*}
Hence, the difference between the expected outcomes under the  optimal rule and treatment 1 is
\begin{align}\label{benefit_form}
\mu_O - E(Y^1) &= \mathrm{pr}(Z=0)E(Y^0-Y^1\mid Z=0)=\mathrm{pr}(\Delta< 0)E(|\Delta|\mid \Delta< 0) = E(|\Delta|I\{\Delta<0\}).
\end{align}
As $E(Y^1)\geq E(Y^0)$, $ 0\le E(\Delta) = E(|\Delta|I\{\Delta\ge 0\}) - E(|\Delta|I\{\Delta<0\})$, and so $E(|\Delta|I\{\Delta<0\})\le E(|\Delta|I\{\Delta\ge 0\})$. Combining this with the fact that $E(|\Delta|)=E(|\Delta|I\{\Delta\ge 0\}) + E(|\Delta|I\{\Delta<0\})$, this shows that $E(|\Delta|)\ge 2E(|\Delta|I\{\Delta<0\})$. Plugging this into the above and applying Jensen's inequality yields that
\begin{align*}
\mu_O - E(Y^1)&\leq \frac{E(|\Delta|)}{2} \leq \frac{1}{2}\sqrt{E(\Delta^2)} = \frac{1}{2}\sqrt{\mathrm{var}(\Delta)+\{E(\Delta)\}^2}
\end{align*}
The result follows from Propositions \ref{prop:sup_var_bound} and \ref{prop:sup_var_bound_bdd} in Appendix \ref{var_bounds}.
\end{proof}

\subsection{Linear optimization bound on rule benefit for bounded outcomes}\label{sup_rule_bound_lp}

We now show that, when the outcome is discrete with finite support, linear programming can be used to derive upper and lower bounds on the benefit of individualizing treatment. These bounds are tight in the sense that, for any observable values of the arm-specific means $E(Y^0)$ and $E(Y^1)$ and arm-specific variances $\mathrm{var}(Y^0)$ and $\mathrm{var}(Y^1)$, there exists a joint distribution of $(Y^0,Y^1)$ for which the marginal distributions of $Y^0$ and $Y^1$ have these means and variances.
\begin{prop}
Suppose that the outcome has support $\{y_1,\ldots,y_k\}$, where $k<\infty$. Fix a joint distribution of $(Y^0,Y^1)$, and let $E(Y^0)$, $E(Y^1)$, $\mathrm{var}(Y^0)$, and $\mathrm{var}(Y^1)$ denote the corresponding arm-specific means and variances. Let $o=\arg\max_{a\in\{0,1\}} E(Y^a)$.  
The optimal value of the following linear programming problem provides a tight upper bound on the benefit of individualizing treatment, namely $\mu_O-\mu_T$:
\begin{eqnarray*}
&\textnormal{maximize} &\sum_{r=1}^k\sum_{j=1}^k (y_r-y_j)p_{rj}I(y_j < y_r) \\[0.5em]
&\textnormal{subject to }
& p_{rj}\ge 0\textnormal{ for all }(r,j)\in\{1,2,\ldots,k\}^2, \\
& &\sum_{r=1}^k\sum_{j=1}^k p_{rj} = 1,\hspace{1em}
\sum_{r=1}^k \sum_{j=1}^k y_j p_{rj} = E(Y^o)\ ,\hspace{1em}
\sum_{r=1}^k\sum_{j=1}^k y_r p_{rj}  = E(Y^{1-o})\ ,\\
& &\sum_{r=1}^k \sum_{j=1}^k y_j^2 p_{rj} = \mathrm{var}(Y^o) + \{E(Y^o)\}^2\ ,\hspace{1em}\sum_{r=1}^k\sum_{j=1}^k y_r^2 p_{rj} =\mathrm{var}(Y^{1-o}) + \{E(Y^{1-o})\}^2\ .
\end{eqnarray*}
The variables in the above linear program are $(p_{rj} : (r,j)\in\{1,2,\ldots,k\}^2)$. 
A tight lower bound on the benefit of individualizing treatment is given by the linear program that is the same as the above, but with the objective minimized instead of maximized.
\end{prop}

\begin{proof}
We only present the proof of the upper bound; the lower bound proof is nearly identical. Let $\mathbb{P}$ denote the true joint distribution of $(Y^0,Y^1)$ that was fixed in the statement of the proposition. 
For clarity, in this proof we will use subscripts to denote which distribution expectations, variances, and probabilities are taken under.

For a generic value $(p_{rj} : (r,j)\in\{1,2,\ldots,k\}^2)$ of the variable used by the linear program, we will let $P$ denote the joint distribution of $(Y^0,Y^1)$ such that $p_{rj}=\mathrm{pr}_P(Y^{1-o}=y_r,Y^o=y_j)$. The feasible region of the linear program consists precisely of those $p_{rj}$ for which the arm-specific means and variances under the corresponding distribution $P$ match those under $\mathbb{P}$, that is, for which $E_P(Y^0)=E_\mathbb{P}(Y^0)$, $E_P(Y^1)=E_\mathbb{P}(Y^1)$, $\mathrm{var}_P(Y^0)=\mathrm{var}_\mathbb{P}(Y^0)$, and $\mathrm{var}_P(Y^1)=\mathrm{var}_\mathbb{P}(Y^1)$. Moreover, by Eq. \ref{benefit_form}, the objective function of the linear program corresponds to the benefit of individualizing treatment when $(Y^0,Y^1)\sim P$. As $\mathbb{p}_{rj}=\mathrm{pr}_{\mathbb{P}}(Y^{1-o}=y_r,Y^o=y_j)$ belongs to the feasible region of the linear program, this implies that the optimal value provides an upper bound on $\mu_O-\mu_T$.

To show that the upper bound is tight, we begin by noting that the feasible region is a subset of the $(k^2-1)$-simplex, and therefore is bounded. Moreover, as $\mathbb{p}_{rj}=\mathrm{pr}_{\mathbb{P}}(Y^{1-o}=y_r,Y^o=y_j)$ belongs to the feasible region, this region is non-empty. As linear programs with bounded, non-empty feasible regions have an optimal solution, we can let $P^\star$ be a distribution such that the objective function evaluated at $p_{rj}^\star=\mathrm{pr}_{P^\star}(Y^{1-o}=y_r,Y^o=y_j)$ is equal to the optimal value. Hence, there exists a distribution $P^\star$ such that the benefit of individualizing treatment is equal to our upper bound when $(Y^0,Y^1)\sim P^\star$.
\end{proof}

\section{Results for the case where stratified summary statistics are available}\label{extra_bounds}

\subsection{Forms of the linear programs}\label{app:stratifiedlp}

We now present the forms of the linear programs used to define the quantities $b_c^+$ and $b_c^-$ discussed in Section~\ref{extensions_section} of the main text. Fix a realization $c$ of $C$. Let $o(c)=\arg\max_{a\in\{0,1\}} E(Y^a\mid C=c)$ denote the optimal treatment within stratum $c$. The quantity $b_{c}^+$ is the optimal value of the linear program
\begin{eqnarray*}\label{lp_bound_stratified}
&\textnormal{maximize} &\sum_{r=1}^k\sum_{j=1}^k (y_r-y_j)p_{rj}I(y_j < y_r) \\[0.5em]
&\textnormal{subject to } & p_{rj}\ge 0\textnormal{ for all }(r,j)\in\{1,2,\ldots,k\}^2, \\ & &\sum_{r=1}^k\sum_{j=1}^k p_{rj} = 1,\hspace{1em}
\sum_{r=1}^k \sum_{j=1}^k y_j p_{rj} = E[Y^{o(c)}\mid C=c]\ ,\hspace{1em}
\sum_{r=1}^k\sum_{j=1}^k y_r p_{rj}  = E[Y^{1-o(c)}\mid C=c]\ ,\\
& &\sum_{r=1}^k \sum_{j=1}^k y_j^2 p_{rj} = \mathrm{var}[Y^{o(c)}\mid C=c] + \{E[Y^{o(c)}\mid C=c)]\}^2\ \\&&\sum_{r=1}^k\sum_{j=1}^k y_r^2 p_{rj} =\mathrm{var}[Y^{1-o(c)}\mid C=c] + \{E[Y^{1-o(c)}\mid C=c]\}^2\ ,
\end{eqnarray*}
and $b^-_{c}$ is the optimal value of the same linear program but with the objective minimized instead of maximized. The quantity $p_{rj}$ above can be thought of as representing a potential value of $\mathrm{pr}(Y^{1-o(c)}=y_r,Y^{o(c)}=y_j\mid C=c)$.

\subsection{Closed-form bounds}\label{sup_strata_bounds}

We prove the general bound on treatment effect heterogeneity incorporating a stratifying variable.

\begin{prop}\label{strata_prop}
Let $C$ be a discrete covariate variable with levels $c_1,...,c_k$. Given lower and upper bounds $B^-_s$ and $B^+_s$ for treatment effect heterogeneity within covariate stratum $C=c_s$, 
the treatment effect heterogeneity, $\mathrm{var}(\Delta)$, can be bounded as
\begin{align}
r_C+\sum_{s=1}^k \mathrm{pr}(C=c_s) B_s^-\leq \mathrm{var}(\Delta)\leq r_C+\sum_{s=1}^k \mathrm{pr}(C=c_s) B_s^+\ , \label{eq:heterogeneityBd}
\end{align}
where we define $r_C=\sum_{s=1}^k \mathrm{pr}(C=c_s)\left\{E(\Delta\mid C=c_s)-E(\Delta)\right\}^2$.
\end{prop}
\begin{proof}
By the law of total variance,
\begin{align*}
    \mathrm{var}(\Delta)&= \mathrm{var}\{E(\Delta\mid C)\} + E\{\mathrm{var}(\Delta \mid C)\}\\
    &=\sum_{s=1}^k \mathrm{pr}(C=c_s) \left[\{E(\Delta\mid C=c_s)-E(\Delta)\}^2+\mathrm{var}(\Delta\mid C=c_s)\right].
\end{align*}
We can then use that $B_s^-\le \mathrm{var}(\Delta\mid C=c_s)\le B_s^+$ to show that the bound is valid.
\end{proof}
Similarly to the bound in Eqs.~\ref{eq:tightVarBd} and \ref{eq:tightVarBdBddY} for the case where stratified summary statistics are not available, the bound in Eq.~\ref{eq:heterogeneityBd} can be shown to be tight when $B_s^-$ and $B_s^+$ correspond to the stratum-$c_s$-specific variance bounds provided in Proposition~\ref{prop:sup_var_bound} (if the counterfactual outcomes are known to have finite variance) or in Proposition~\ref{prop:sup_var_bound_bdd} (if the counterfactual outcomes are known to fall in $(m,M)$).

\section{Inference for the bounds}\label{sup_inference}

\subsection{Inference for the closed-form bounds}

To quantify the uncertainty of our closed-form bounds, we will derive their limiting variances in the asymptotic regime where the sample sizes of the arms in the clinical trial go to infinity. 
We begin by presenting the form of these variances. 
If consistent estimates of these variances are available, then standard Wald-type confidence intervals for our closed-form bounds can be constructed. 
We will provide the form of the intervals in this case. 
Unfortunately, consistent estimates of these variances will often not be available in the setting we mainly focus on in this paper, namely the one where individual-level data are not available. 
For these cases, we derive upper bounds on the limiting variances that are based only on these summary statistics and known bounds on the outcome. 
We then describe how to use these upper bounds to construct conservative confidence intervals.

We will allow for the treatment and control arms to have different sample sizes, namely $n_1$ and $n_0$, and assume that the ratios $r_1=n_1/(n_1+n_0)$ and $r_0=n_0/(n_1+n_0)$ converge to $(0,1)$-valued constants $\rho_1$ and $\rho_0$ as the overall sample size increases. 
Let $\mu_{4,a}=E[\{Y^a-E(Y^a)\}^4]$ be the fourth population central moment of arm $a$, which we assume is finite. We also let $\bar{\Delta}$ denote the sample average treatment effect and $s_0^2$ and $s_1^2$ denote the estimates of $\mathrm{var}(Y^0)$ and $\mathrm{var}(Y^1)$, respectively. The estimate of the variance ratio $\nu^2$ is given by $\hat{\nu}^2=s_1^2/s_0^2$. 

The estimated general upper and lower closed-form bounds on ${\rm var}(\Delta)$ from Eq.~\ref{teh_bound} are given by $s_0^2(\hat{\nu}- 1)^2$ and $s_0^2(\hat{\nu}+ 1)^2$, respectively. By the delta method, these bounds satisfy the following convergence in distribution results:
\begin{align*}
\sqrt{n_1+n_0}\{s_0^2(\hat{\nu}- 1)^2 -\mathrm{var}(Y^0)(\nu- 1)^2\}&\cd N(0,\tau_{\mathrm{het}}^-), \\
\sqrt{n_1+n_0}\{s_0^2(\hat{\nu}+ 1)^2 -\mathrm{var}(Y^0)(\nu+ 1)^2\}&\cd N(0,\tau_{\mathrm{het}}^+),
\end{align*}
where
\begin{align*}
    \tau_{\mathrm{het}}^-&= \left(\frac{\mu_{4,1}}{\rho_1}-\mathrm{var}(Y^1)^2\right)(\nu^{-1}- 1)^2+\left(\frac{\mu_{4,0}}{\rho_0}-\mathrm{var}(Y^0)^2\right)(\nu- 1)^2, \\
    \tau_{\mathrm{het}}^+&= \left(\frac{\mu_{4,1}}{\rho_1}-\mathrm{var}(Y^1)^2\right)(\nu^{-1}+ 1)^2+\left(\frac{\mu_{4,0}}{\rho_0}-\mathrm{var}(Y^0)^2\right)(\nu+ 1)^2.
\end{align*}
These convergence in distribution results facilitate the construction of Wald-type confidence intervals for the upper and lower bounds on ${\rm var}(\Delta)$ from Eq.~\ref{teh_bound}. In particular, the following interval will contain these upper and lower bounds, and therefore ${\rm var}(\Delta)$, with probability tending to $1-\alpha$:
\begin{equation}\label{eq:ci}
  \left[s_0^2(\hat{\nu}- 1)^2 - z_{1-\alpha/2}\sqrt{\frac{\tau_{\mathrm{het}}^-}{n_1+n_0}} ,\, s_0^2(\hat{\nu}+ 1)^2 + z_{1-\alpha/2}\sqrt{\frac{\tau_{\mathrm{het}}^+}{n_1+n_0}}\right],
\end{equation}
where throughout, for $\beta\in (0,1)$, $z_\beta$ denotes the $\beta$-percentile of the standard normal distribution. Since the bounds in Eq.~\ref{teh_bound_bounded} are at least as tight as those in Eq.~\ref{teh_bound} when the outcome is bounded, the above is also an asymptotically valid confidence interval for the bounds in Eq.~\ref{teh_bound_bounded} in this case. When $\tau_{\mathrm{het}}^-$ and $\tau_{\mathrm{het}}^+$ are not known, an asymptotically $(1-\alpha)$-level confidence interval would still be obtained if they were replaced by consistent estimates, and an asymptotically conservative confidence interval would be obtained if they were instead replaced by asymptotic upper bounds.

In the case where no bounds on the outcome are known, the estimated upper bound on the benefit of individualizing from Eq.~\ref{rule_bound} is given by $\frac{1}{2}\sqrt{s_0^2(\hat{\nu}+1)^2 +\bar{\Delta}^2}$. 
By the delta method, it holds that
\begin{align*}
\sqrt{n_1+n_0}\left(\frac{1}{2}\sqrt{s_0^2(\hat{\nu}+1)^2 +\bar{\Delta}^2}-\frac{1}{2}\sqrt{\mathrm{var}(Y^0)(\nu+1)^2 +E(\Delta)^2}\right) \cd N\left(0,\tau_{\mathrm{ben}}\right),
\end{align*}
where
\begin{align*}
    \tau_{\mathrm{ben}} = &\frac{1}{16\{\mathrm{var}(Y^0)(\nu+1)^2 +E(\Delta)^2\}}\cdot \Bigg[\left(\frac{\mu_{4,1}}{\rho_1}-\mathrm{var}(Y^1)^2\right) (1+\nu^{-1})^2\\&+\left(\frac{\mu_{4,0}}{\rho_0}-\mathrm{var}(Y^0)^2\right)(1+\nu)^2+4\left(\frac{\mathrm{var}(Y^1)+E(Y^1)^2}{\rho_1}+\frac{\mathrm{var}(Y^0)+E(Y^0)^2}{\rho_0}-E(\Delta)^2\right)E(\Delta)^2
    \\&+4\left(\frac{E((Y^1)^3)-2E(Y^1)\mathrm{var}(Y^1)-E(Y^1)^3}{\rho_1}-E(\Delta)\mathrm{var}(Y^1)\right)E(\Delta)(1+\nu^{-1})\\
    &-4\left(\frac{E((Y^0)^3)-2E(Y^0)\mathrm{var}(Y^0)-E(Y^0)^3}{\rho_0}+E(\Delta)\mathrm{var}(Y^0)\right)E(\Delta)(1+\nu)\Bigg].
\end{align*}
An asymptotic $(1-\alpha)$-level upper-confidence bound for the right-hand side of Eq.~\ref{rule_bound} when $B^+$ is given by Eq.~\ref{teh_bound}, and an asymptotically conservative $(1-\alpha)$-level bound when $B^+$ is given by Eq.~\ref{teh_bound_bounded}, is given by
\begin{equation}\label{eq:ciUB}
 \frac{1}{2}\sqrt{s_0^2(\hat{\nu}+1)^2 +\bar{\Delta}^2} + z_{1-\alpha}\sqrt{\frac{\tau_{\mathrm{ben}}}{n_1+n_0}}.
\end{equation}
Here we are interested only in an upper confidence bound as our closed-form bound is itself an upper bound. 

Since the asymptotic variances $\tau_{\mathrm{het}}^-$, $\tau_{\mathrm{het}}^+$, and $\tau_{\mathrm{ben}}$ depend on true variances and moments of the outcome, their values will not generally be known. However, in the situation where individual-level treatment assignments and outcomes are accessible (even if covariates are not), these variances can be consistently estimated by plugging in empirical variances and moments into their definitions.

Consistent estimation of $\tau_{\mathrm{het}}^-$, $\tau_{\mathrm{het}}^+$, and $\tau_{\mathrm{ben}}$ is not generally possible when individual-level data is not available. To enable inference in these cases, we will derive upper bounds on these variances that can be consistently estimated using only commonly reported summary statistics. Inspecting the form of these variances, we see that we need to estimate or bound $\nu$, $\mathrm{var}(Y^a)$, $E(Y^a)$, and $\mu_{4,a}$ for $a\in\{0,1\}$. We can estimate $\mathrm{var}(Y^a)$ with $s_a^2$ and $\nu$ with $\hat{\nu}=s_1/s_0$. We can estimate $E(Y^a)$ with $\bar{Y}^a$, which we use to denote the empirical mean of $Y^a$. If the outcome is unbounded, we likely have no information which we could use to bound or estimate $\mu_{4,a}$. However, if the outcome is known to be bounded within $[m,M]$, we can use an existing result \cite{sharma2015complementary} to show that, for $a\in\{0,1\}$,
\begin{equation}
    \mu_{4,a}\leq \max\{[M-E(Y^a)]^2,[m-E(Y^a)]^2\}\mathrm{var}(Y^a). \label{eq:momentBd}
\end{equation}

The above can be used to show that $\tau_{\textnormal{het}}^-\le f(-1)$ and $\tau_{\textnormal{het}}^+\le f(1)$, where
\begin{align*}
    f(z)&=\left(\frac{\max\{[M-E(Y^1)]^2,[m-E(Y^1)]^2\}\mathrm{var}(Y^1)}{\rho_1}-\mathrm{var}(Y^1)^2\right)(\nu^{-1}+ z)^2 \\
    &\quad+\left(\frac{\max\{[M-E(Y^0)]^2,[m-E(Y^0)]^2\}\mathrm{var}(Y^0)}{\rho_0}-\mathrm{var}(Y^0)^2\right)(\nu+z)^2.
\end{align*}
We can consistently estimate $f(z)$ with the following, which can be evaluated using only summary statistics and the known bounds on the outcome:
\begin{align*}
    \hat{f}(z)&= 
    \left(\frac{\max\{(M-\bar{Y}^1)^2,(m-\bar{Y}^1)^2\}s_1^2}{r_1}-s_1^4\right)(\hat{\nu}^{-1}+ z)^2 \\
    &\quad+\left(\frac{\max\{(M-\bar{Y}^0)^2,(m-\bar{Y}^0)^2\}s_0^2}{r_0}-s_0^4\right)(\hat{\nu}+ z)^2,
\end{align*}
An asymptotically conservative $(1-\alpha)$-level confidence interval for $\mathrm{var}(\Delta)$ is obtained by replacing $\tau_{\mathrm{het}}^-$ and $\tau_{\mathrm{het}}^+$ in Eq.~\ref{eq:ci} by $\hat{f}(-1)$ and $\hat{f}(1)$, respectively.

Next, turning to $\tau_{\mathrm{ben}}$, the asymptotic variance of the closed-form bound on the benefit of individualizing treatment, we see that this variance involves third-moment terms. Similarly to Eq. \ref{eq:momentBd} for outcomes known to be bounded in $[m,M]$, we can bound these terms as
\begin{align}
|E[(Y^a)^3]|&\leq \max(|M|,|m|)E[(Y^a)^2] = \max(|M|,|m|)\{\mathrm{var}(Y^a)+E(Y^a)^2\}. \label{eq:thirdMomentBd}
\end{align}
Then, we can consistently estimate an upper bound on $\tau_{\mathrm{ben}}$ using the following:

\begin{align*}
    &\frac{1}{16\{s_0^2(\hat{\nu}+1)^2 +\bar{\Delta}^2\}}\cdot \Bigg[\left(\frac{\max\{(M-\bar{Y}^1)^2,(m-\bar{Y}^1)^2\}s_1^2}{r_1}-s_1^4\right) (1+\hat{\nu}^{-1})^2\\&+\left(\frac{\max\{(M-\bar{Y}^0)^2,(m-\bar{Y}^0)^2\}s_0^2}{r_0}-s_0^4\right)(1+\nu)^2+4\left(\frac{s_1^2+(\bar{Y}^1)^2}{r_1}+\frac{s_0^2+(\bar{Y}^0)^2}{r_0}-\bar{\Delta}^2\right)\bar{\Delta}^2
    \\&+4\left(\frac{\max(|M|,|m|)\{s_1^2+(\bar{Y}^1)^2\}-2\bar{Y}^1s_1^2-(\bar{Y}^1)^3}{r_1}-\bar{\Delta} s_1^2\right)\bar{\Delta}(1+\hat{\nu}^{-1})\\
    &+4\left(\frac{\max(|M|,|m|)\{s_0^2+(\bar{Y}^0)^2\}+2\bar{Y}^0 s_0^2+(\bar{Y}^0)^3}{r_0}-\bar{\Delta} s_0^2\right)\bar{\Delta}(1+\hat{\nu})\Bigg].
\end{align*}
An asymptotically conservative $(1-\alpha)$-level bound on the benefit of individualizing treatment, namely $\mu_O-\mu_T$, can be obtained by replacing $\tau_{\mathrm{ben}}$ in Eq.~\ref{eq:ciUB} by the above.

\subsection{Confidence intervals for linear programming bound on treatment benefit}

Let $(1-\alpha)\in(0,1)$ be the desired confidence level. To provide a confidence interval for the linear programming treatment benefit bounds, we can first notice that the uncertainty in the bounds comes from the uncertainty in the estimates of the optimization constraints. To address this, we can change our optimization problem so that instead of the optimal solution producing arm-specific means and second moments equal to our estimates, we require the optimal solution to produce arm-specific means and second moments within $(1-\alpha/4)$-level confidence intervals for the population-level means and second moments. Applying a union bound, these four confidence intervals will all contain the true arm-specific means and second moments with probability tending to at least $(1-\alpha)$. Hence, asymptotically there will be at least a probability of $1-\alpha$ that the optimal value found by maximizing (minimizing) the linear program presented in Eq.~\ref{lp_bound} of the main text falls below (above) that obtained by maximizing (minimizing) the linear program using these wider constraints. Taking these two bounds on the optimal values from the linear program in Eq.~\ref{lp_bound} as the extreme values of our confidence interval, we will obtain an interval that contains the bounds on the benefit of individualizing treatment with asymptotic probability at least $1-\alpha$.

We now provide the forms of the four $(1-\alpha/4)$-level confidence intervals described above. For each $a\in\{0,1\}$, the interval for $E(Y^a)$ takes the form $\bar{Y}^a \pm z_{1-\alpha/8}\cdot s_a^2/\sqrt{n_a}$. When constructing the intervals for the second moments $E((Y^a)^2)$, $a\in \{0,1\}$, we assume that $Y$ is bounded and known to fall in $[m,M]$. The $(1-\alpha/8)$-level confidence interval for $E((Y^a)^2)$ takes the form $s_a^2+(\bar{Y}^a)^2\pm z_{1-\alpha/8}\cdot \gamma_a/\sqrt{n_a}$, where
\begin{align*}
    \gamma_a = \left(\max\{(M-\bar{Y}^a)^2,(m-\bar{Y}^a)^2\}s_a^2-s_a^4 + 4|\bar{Y}^a|\max(|m|,|M|)(s_a^2+(\bar{Y}^a)^2)-8s_a^2(\bar{Y}^a)^2 - 4(\bar{Y}^a)^4\right)^{1/2}.
\end{align*}
Though we omit the details here, Eqs. \ref{eq:momentBd} and \ref{eq:thirdMomentBd} can be used to show that $\gamma_a^2$ consistently estimates an upper bound on $\mathrm{var}((Y^a)^2)$.

We now provide the form of a two-sided $(1-\alpha)$-level confidence interval $[L,U]$ that will contain the two-sided bound on the benefit of individualizing treatment that can be obtained from the linear program presented in Eq.~\ref{lp_bound}. When doing so, we assume, without loss of generality, that $E(Y^1)\ge E(Y^0)$. The upper bound $U$ is given by the optimal value of the following linear program:
\begin{align*}
&{\rm maximize }\ \ \sum_{r=1}^k\sum_{j=1}^k (y_r-y_j)p_{rj}I(y_j < y_r)\textnormal{ subject to }\\[0.5em]
\ \ &p_{rj}\ge 0\textnormal{ for all }(r,j)\in\{1,2,\ldots,k\}^2,\hspace{1em}\sum_{r=1}^k\sum_{j=1}^k p_{rj} = 1, \\
\ \ &\sum_{r=1}^k \sum_{j=1}^k y_j p_{rj} \ge \bar{Y}^1 - z_{1-\alpha/8}\cdot s_1^2/\sqrt{n_1},\hspace{1em} \sum_{r=1}^k \sum_{j=1}^k y_j p_{rj} \le \bar{Y}^1 + z_{1-\alpha/8}\cdot s_1^2/\sqrt{n_1},\\
\ \ &\sum_{r=1}^k\sum_{j=1}^k y_r p_{rj}  \ge \bar{Y}^0 - z_{1-\alpha/8}\cdot s_0^2/\sqrt{n_0},\hspace{1em} \sum_{r=1}^k\sum_{j=1}^k y_r p_{rj}  \le \bar{Y}^0 + z_{1-\alpha/8}\cdot s_0^2/\sqrt{n_0}, \\
&\sum_{r=1}^k \sum_{j=1}^k y_j^2 p_{rj} \ge s_1^2+(\bar{Y}^1)^2- z_{1-\alpha/8}\cdot \gamma_1/\sqrt{n_1},\hspace{1em} \sum_{r=1}^k \sum_{j=1}^k y_j^2 p_{rj} \le s_1^2+(\bar{Y}^1)^2+ z_{1-\alpha/8}\cdot \gamma_1/\sqrt{n_1}, \\
&\sum_{r=1}^k\sum_{j=1}^k y_r^2 p_{rj} \ge s_0^2 + (\bar{Y}^0)^2 - z_{1-\alpha/8}\cdot \gamma_0/\sqrt{n_0},\hspace{1em} \sum_{r=1}^k\sum_{j=1}^k y_r^2 p_{rj} \le s_0^2 + (\bar{Y}^0)^2 + z_{1-\alpha/8}\cdot \gamma_0/\sqrt{n_0}.
\end{align*}
The lower bound $L$ is given by the optimal value of the same linear program but with the objective being minimized instead of maximized.

\section{Details of optimal rule benefit estimation in the EMBARC example}\label{sup_data_example}

To estimate the benefit of the optimal rule for the EMBARC data, we used all numerical and categorical covariates. For rating scale covariates, we used only total and subtotal variables and excluding variables for individual questions. We created one dataset including all variables with no more than 25\% missingness, and one dataset including all variables with no more than 50\% missingness. The variables used in the analysis with no more than 25\% missingness are listed in Table \ref{tab_25}. The additional variables used in the analysis with no more than 50\% missingness are listed in Table \ref{tab_50}. Missing covariates were singly imputed using the \texttt{MICE} package in R version 4.0.5 \cite{mice, r} with the package's default imputation methods.

To estimate the benefit of the optimal treatment rule, we used cross validated targeted maximum likelihood estimation and \texttt{superlearner} as implemented in the \texttt{sg} package in R \cite{luedtke2016super,sg,sl}. Within \texttt{superlearner}, we used the \texttt{polymars}, \texttt{ranger}, \texttt{nnet}, stepwise regression, and stepwise regression with interactions \cite{polymars,ranger,nnet}. For all methods, we initially screened variables using a random forest \cite{randomforest}. We excluded all observations with missing outcome measurements to match the summary statistics, which only included complete cases. We used the sg.cvtlme function with default settings.

Using the dataset with no more than 25\% missingness, we estimated a 95\% confidence interval of (0, 0.09) for the benefit. Using the dataset with no more than 50\% missingness, our estimated confidence interval was (0, 0.08).

\begin{table}
\begin{small}
\caption{Covariates used in the analysis with no more than 25\% missingness}
\label{tab_25}
\begin{tabular}{|l|l|l|}
\hline
 \textbf{Variable} & \textbf{Form} & \textbf{Description}\\
\hline
aaq\_1 & Anger Attack Questionnaire & Past 6 months, have felt irritable or easily angered\\
\hline
aaq\_2 & Anger Attack Questionnaire & Past 6 months feel have overreacted with anger or rage to \\
&  & \hspace{2 mm} minor annoyances or trivial issues\\
\hline
aaq\_3 & Anger Attack Questionnaire & Past 6 months, had "anger attacks", where would become\\
&  & \hspace{2 mm} angry \& enraged with others in way thought was excessive\\
\hline
aaq\_score\_result & Anger Attack Questionnaire & presence of anger attack\\
\hline
asrm\_score2 & Altman Self-Rating Mania Scale & ASRM score\\
\hline
atrq\_01 & Antidepressant Trt. Response Q & Received medication treatment since beginning of current \\
&  & \hspace{2 mm} episode or last 2 years if current episode lasted over 2 yrs\\
\hline
cast\_score\_total & Concise Assoc. Symptoms Scale & Total Score\\
\hline
cgi\_si & Clinical Global Impression & Severity of Illness. Considering total clinical experience with\\
 &  & \hspace{2 mm} this particular population how ill is subject at this time\\
\hline
chrtp\_propensity\_score & Concise Health Risk Tracking & Propensity score\\
\hline
chrtp\_risk\_score & Concise Health Risk Tracking & Risk score\\
\hline
ctqscore\_ea & Childhood Trauma Q & Emotional Abuse\\
\hline
ctqscore\_en & Childhood Trauma Q & Emotional Neglect\\
\hline
ctqscore\_pa & Childhood Trauma Q & Physical Abuse\\
\hline
ctqscore\_pn & Childhood Trauma Q & Physical Neglect\\
\hline
ctqscore\_sa & Childhood Trauma Q & Sexual Abuse\\
\hline
ctqscore\_val & Childhood Trauma Q & Validity\\
\hline
ehiscore & Edinburgh Handedness Invent. & EHI Score\\
\hline
interview\_age & Demographics & Age in months.\\
\hline
sex & Demographics & Sex of subject at birth\\
\hline
demo\_educa\_status & Demographics & Education\\
\hline
edutot & Demographics & Total years of education\\
\hline
employst & Demographics & Current employment status\\
\hline
thous & Demographics & Total number of persons in household\\
\hline
demo\_maritial\_status & Demographics & Marital Status:\\
\hline
totincom & Demographics & Monthly household income\\
\hline
demo\_ssn\_status & Demographics & Social Security Number Status (not the number)\\
\hline
hispanic & Demographics & Is subject of Hispanic, Latino, or Spanish origin?\\
\hline
race & Demographics & Race of study subject\\
\hline
mhf\_23 & Medical History Form & Checks if physical exam was administered\\
\hline
completed & Medical History Form & Checks if completed\\
\hline
fhs\_01\_kids & Family History Screen Modified & Number biological children know about\\
\hline
fhs\_01\_parent & Family History Screen Modified & Number biological parents know about\\
\hline
fhs\_01\_sibs & Family History Screen Modified & Number biological brothers/sisters know about\\
\hline
fhs\_02 & Family History Screen Modified & Any biological children/parents/siblings ever had serious \\
 &  & \hspace{2 mm} mental illness, emotional problem, or  nervous breakdown\\
\hline
fhs\_03 & Family History Screen Modified & Any biological children/parents/siblings in past 2 weeks feel \\
 &  & \hspace{2 mm} sad or blue\\
\hline
fhs\_04 & Family History Screen Modified & Any biological children/parents/siblings ever had period of \\
 &  & \hspace{2 mm} feeling extremely happy or high\\
\hline
fhs\_05 & Family History Screen Modified & Any biological children/parents/siblings ever heard voices, \\
 &  & \hspace{2 mm} or seen visions, that other people could not see or hear\\
\hline
fhs\_06 & Family History Screen Modified & Any biological children/parents/siblings ever tried to kill \\
 &  & \hspace{2 mm} him or herself, or made a suicide attempt\\
\hline
fhs\_07 & Family History Screen Modified & Any biological children/parents/siblings commit suicide\\
\hline
hamd\_36 & Hamilton Depression Scale (32) & HAMD Total 17 Item Score\\
\hline
hamd\_score\_24 & Hamilton Depression Scale (32) & HAMD Total 24 Item Score\\
\hline
\end{tabular}
\end{small}
\end{table}

\begin{table}
\begin{small}
\begin{tabular}{|l|l|l|}
\hline
 \textbf{Variable} & \textbf{Form} & \textbf{Description}\\
\hline
masq2\_score\_aa & Mood and Anxiety Symptoms Q & Anxious Arousal\\
\hline
masq2\_score\_ad & Mood and Anxiety Symptoms Q & Anhedonic Depression\\
\hline
masq2\_score\_gd & Mood and Anxiety Symptoms Q & General Distress Anxious Symptoms\\
\hline
mdqscore\_total & Mood Disorders Q & sums the first thirteen q's\\
\hline
sfi\_total & MGH Sexual Function Inventory & SFI total score\\
\hline
neo\_n & NEO-Five Factor Inventory & Neuroticism raw score\\
\hline
neo\_e & NEO-Five Factor Inventory & Extraversion raw score\\
\hline
neo\_o & NEO-Five Factor Inventory & Openness to Experience raw score\\
\hline
neo\_a & NEO-Five Factor Inventory & Agreeableness raw score\\
\hline
neo\_c & NEO-Five Factor Inventory & Conscientiousness raw score\\
\hline
sapas\_score & SAP Abbrev. Scale & SAPAS total\\
\hline
q001\_bpi\_life & SCID I/p Summary & Bipolar I Lifetime Prevalence\\
\hline
q010\_bpii\_life & SCID I/p Summary & Bipolar II Lifetime Prevalence\\
\hline
q018\_obp\_life & SCID I/p Summary & Other Bipolar Lifetime Prevalence\\
\hline
q021\_mdd\_life & SCID I/p Summary & Major Depressive Disorder Lifetime Prevalence\\
\hline
q022\_mdd\_ep & SCID I/p Summary & Major Depressive Disorder Episode Pattern\\
\hline
q026\_mdd\_current & SCID I/p Summary & Major Depressive Disorder Current Episode\\
\hline
q027\_mdd\_sev & SCID I/p Summary & Major Depressive Disorder Current Severity\\
\hline
q065\_pdnos\_life & SCID I/p Summary & Psychotic Disorder NOS Lifetime Prevalence\\
\hline
q067\_al\_life & SCID I/p Summary & Alcohol Lifetime Prevalence\\
\hline
q069\_sha\_life & SCID I/p Summary & Sedative-Hypnotic-Anxiolytic Lifetime Prevalence\\
\hline
q071\_can\_life & SCID I/p Summary & Cannabis Lifetime Prevalence\\
\hline
q073\_stim\_life & SCID I/p Summary & Stimulants Lifetime Prevalence\\
\hline
q075\_op\_life & SCID I/p Summary & Opiod Lifetime Prevalence\\
\hline
q077\_coc\_life & SCID I/p Summary & Cocaine Lifetime Prevalence\\
\hline
q079\_hal\_life & SCID I/p Summary & Hallucinogen/PCP Lifetime Prevalence\\
\hline
q081\_poly\_life & SCID I/p Summary & Polysubstance Lifetime Prevalence\\
\hline
q083\_othsub\_life & SCID I/p Summary & Other Substance Lifetime Prevalence\\
\hline
q085\_panic\_life & SCID I/p Summary & Panic Disorder Lifetime Prevalence\\
\hline
q088\_agor\_life & SCID I/p Summary & Agoraphobia without Panic Disorder Lifetime Prevalence\\
\hline
q090\_social\_life & SCID I/p Summary & Social Phobia Lifetime Prevalence\\
\hline
q092\_phobia\_life & SCID I/p Summary & Specific Phobia Lifetime Prevalence\\
\hline
q094\_ocd\_life & SCID I/p Summary & Obsessive Compulsive Disorder Lifetime Prevalence\\
\hline
q096\_ptsd\_life & SCID I/p Summary & Posttraumatic Stress Disorder Lifetime Prevalence\\
\hline
q103\_siad\_life & SCID I/p Summary & Substance-Induced Anxiety Disorder Lifetime Prevalence\\
\hline
q107\_adnos\_life & SCID I/p Summary & Anxiety Disorder NOS Lifetime Prevalence\\
\hline
q114\_an\_life & SCID I/p Summary & Annorexia Nervosa Lifetime Prevalence\\
\hline
q116\_bn\_life & SCID I/p Summary & Bulimia Nervosa Lifetime Prevalence\\
\hline
q121\_othdx\_life & SCID I/p Summary & Other DSM IV Disorder Lifetime Prevalence\\
\hline
gadc & SCID I/p Summary & Generalized Anxiety Disorder current\\
\hline
mdca4 & SCID I/p Summary & Insomnia or hypersomnia\\
\hline
mdca6 & SCID I/p Summary & Fatigue or loss of energy\\
\hline
mdca7 & SCID I/p Summary & Feelings of worthlessness or inappropriate guilt\\
\hline
mdca8 & SCID I/p Summary & Trouble thinking or concentrating\\
\hline
mdca9 & SCID I/p Summary & Thoughts or plans of death\\
\hline
mdcnumep & SCID I/p Summary & Total episodes of major depressive syndrome, including current\\
\hline
mdcageon & SCID I/p Summary & Age at onset of first unequivocal major depressive syndrome\\
\hline
dsm4oma1 & SCID I/p Summary & Prominent and persistent depressed mood or diminished interest\\
\hline
dsm4mb5 & SCID I/p Summary & Significant anorexia or weight loss\\
\hline
dsm4mb6 & SCID I/p Summary & Excessive or inappropriate guilt\\
\hline
scid\_insomnia & SCID I/p Summary & Insomnia\\
\hline
\end{tabular}
\end{small}
\end{table}

\begin{table}
\begin{small}
\begin{tabular}{|l|l|l|}
\hline
 \textbf{Variable} & \textbf{Form} & \textbf{Description}\\
\hline
psyret & SCID I/p Summary & Psychomotor retardation\\
\hline
pla5 & SCID I/p Summary & Major Depressive Episode - Psychomotor agitation\\
\hline
worthled & SCID I/p Summary & Worthlessness\\
\hline
pla8 & SCID I/p Summary & MD Episode - Diminished ability to think\\
\hline
indecide & SCID I/p Summary & MD Episode - Indecisiveness\\
\hline
suicidea & SCID I/p Summary & MD Episode - Suicidal Ideation\\
\hline
specplan & SCID I/p Summary & MD Episode - Specific plan to commit suicide\\
\hline
owndeath & SCID I/p Summary & MD Episode - Thoughts of own death\\
\hline
adsq12\_u & SCID I/p Summary & Presence of Symptom: Weight gain (Non-imputed)\\
\hline
adsq10\_u & SCID I/p Summary & Presence of Symptom: Weight loss (Non-imputed)\\
\hline
ca601 & SCID I/p Summary & Past psychosis - MDD: Hypersomnia - depressive ep.\\
\hline
scid\_ad1 & SCID I/p Summary & Age at on set of first dysphoria:\\
\hline
scid\_ad3 & SCID I/p Summary & Duration of current MD episode (months)\\
\hline
scid\_ad5 & SCID I/p Summary & Longest period of w/out dysphoria\\
\hline
scid\_ad\_dm1 & SCID I/p Summary & Full interepisodic recovery\\
\hline
scid\_ad6 & SCID I/p Summary & Rate over all course of dysphoria\\
\hline
scid\_ad7 & SCID I/p Summary & Addendum - depressed mood most of the day\\
\hline
scid\_ad11 & SCID I/p Summary & Addendum - pyschomotor or retardation\\
\hline
scid\_ad16 & SCID I/p Summary & Scores past 9 qs: 5 must be 3 and one of 5 is 7 or 8\\
\hline
scidad\_03\_dx & SCID I/p Summary & Current MDD episode Anxious Distress\\
\hline
scidad\_03\_dx1 & SCID I/p Summary & Current MDD episode Mixed Features\\
\hline
shaps\_total\_continuous & S-H Pleasure Scale & sums all scores ordinal\\
\hline
sas\_overall\_mean & Social Adjustment Scale Short & Overall mean\\
\hline
stai\_eeg\_final\_score & Speilberger State Anxiety Invent. & EEG STAI Final Score\\
\hline
strf\_01 & Screening Test Results & Checks for positive drug test\\
\hline
strf\_09 & Screening Test Results & Checks for other abnormalities\\
\hline
weight\_std & Screening Test Results & Weight - Standard Unit\\
\hline
height\_std & Screening Test Results & Height - Standard Unit\\
\hline
bmi & Screening Test Results & body mass index of subject\\
\hline
strf\_cholesterol\_hdl & Screening Test Results & Cholesterol HDL\\
\hline
strf\_cholesterol\_ldl & Screening Test Results & Cholesterol LDL\\
\hline
strf\_cholesterol\_total & Screening Test Results & Cholesterol total\\
\hline
strf\_cholesterol\_triglycerides & Screening Test Results & Cholesterol triglycerides\\
\hline
strf\_crp & Screening Test Results & CRP\\
\hline
strf\_fasting & Screening Test Results & fasting y/n\\
\hline
gluval & Screening Test Results & glucose value\\
\hline
tsh & Screening Test Results & Thyroid-stimulating Hormone (TSH)\\
\hline
strf\_waist & Screening Test Results & waist circumference\\
\hline
vas1 & Visual Analog Mood Scales & Friendly (P). VAMS Friendly - Hostile\\
\hline
vams\_hap1 & Visual Analog Mood Scales & Happy (P). VAMS Happy - Sad\\
\hline
vams\_wit1 & Visual Analog Mood Scales & Quick Witted (P). VAMS Quick Witted\\
\hline
vams\_rel1 & Visual Analog Mood Scales & Relaxed (P). VAMS Relaxed - Tense\\
\hline
vams\_soc1 & Visual Analog Mood Scales & Sociable (P).VAMS Sociable\\
\hline
ss\_vocabularyrawscore & Wechsler Abbrev. Scale of Intell. & Vocabulary Raw - Score\\
\hline
ss\_matrixreasoningrawscore & Wechsler Abbrev. Scale of Intell. & Matrix Reasoning Raw - Score\\
\hline
wasi001b & Wechsler Abbrev. Scale of Intell. & Vocabulary: T Score\\
\hline
wasi004b & Wechsler Abbrev. Scale of Intell. & Matrix Reasoning: T Score\\
\hline
\end{tabular}
\end{small}
\end{table}

\begin{table}
\caption{Additional covariates used in the analysis with no more than 50\% missingness}
\label{tab_50}
\begin{tabular}{|l|l|l|}
\hline
 \textbf{Variable} & \textbf{Form} & \textbf{Description}\\
\hline
mhf\_02 & Medical History Form & Checks if condition is ongoing\\
\hline
mhf\_04 & Medical History Form & Checks if condition is ongoing - respiratory\\
\hline
mhf\_06 & Medical History Form & Checks if condition is ongoing - circulatory\\
\hline
mhf\_08 & Medical History Form & Checks if condition is ongoing - digestive\\
\hline
mhf\_10 & Medical History Form & Checks if condition is ongoing - skin\\
\hline
mhf\_12 & Medical History Form & Checks if condition is ongoing - urinary\\
\hline
mhf\_14 & Medical History Form & Checks if condition is ongoing - musculoskeletal\\
\hline
mhf\_16 & Medical History Form & Checks if condition is ongoing - nervous\\
\hline
mhf\_18 & Medical History Form & Checks if condition is ongoing - endocrine\\
\hline
mhf\_20 & Medical History Form & Checks if condition is ongoing - blood organs\\
\hline
mhf\_22 & Medical History Form & Checks if condition is ongoing - other III\\
\hline
fhs\_02\_kids & Family History Screen Modified & number of kids with serious mental illness\\
\hline
fhs\_02\_par & Family History Screen Modified & number parents with serious mental illness\\
\hline
fhs\_02\_sibs & Family History Screen Modified & number of siblings with serious mental illness\\
\hline
fhs\_03\_kids & Family History Screen Modified & Number of children feeling sad or blue\\
\hline
fhs\_03\_par & Family History Screen Modified & Number of parents feeling sad or blue\\
\hline
fhs\_03\_sibs & Family History Screen Modified & Number of siblings feeling sad or blue\\
\hline
q025\_mdd\_mon & SCID I/p Summary & Major Depressive Disorder Past Month\\
\hline
scidad\_02\_dx3 & SCID I/p Summary & Current MDD episode type\\
\hline
pregnan & Screening Test Results & pregnancy test\\
\hline
strf\_15 & Screening Test Results & Initial pregnancy test date\\
\hline
\end{tabular}

\end{table}

\end{document}